\documentclass{llncs}
\usepackage{amsfonts,amssymb,amsmath}

\title{Disjunctive Complexity}

\author{N. Ivanov\inst{3}
  \and A. Rubtsov\thanks{The work was done within the framework of the HSE University Basic Research Program.}\inst{2}\inst{3}
  \and
M. Vyalyi\thanks{The work was done within the framework of the HSE University Basic Research Program. Supported  in part by the state assignment topic no. FFNG-2024-0003.}
\inst{1}\inst{2}\inst{3}}
\institute{Federal Research Center “Computer Science and Control” of the Russian Academy of Sciences\\42 Vavilov Street, 119333 Moscow, Russia\\
\and Moscow Institute of Physics and Technology (National Research University),\\1A, bld. 1  Kerchenskaya str, 117303, Moscow, Russia\\
\and National Research University Higher School of Economics,\\ 11 Pokrovsky Blvd, 109028 Moscow, Russia\\
\email{nikita.iva47@gmail.com}\\
\email{rubtsov99@gmail.com}\\
\email{vyalyi@gmail.com}
}

\date{}
\spnewtheorem{Def}{Definition}{\bfseries}{\upshape}
\spnewtheorem{prop}{Proposition}{\bfseries}{\itshape}
\spnewtheorem{Claim}{Claim}{\bfseries}{\itshape}
\spnewtheorem*{known}{Theorem}{\bfseries\upshape}{\itshape}
\spnewtheorem*{Lemma}{Lemma}{\bfseries\upshape}{\itshape}

\let\al\alpha

\let\es\varnothing

\renewcommand*\P{\ensuremath{\mathrm {P}}}
\newcommand*\NP{\ensuremath{\mathrm {NP}}}

\def\ZZ{\mathbb Z}

\def\A{\ensuremath{\mathcal A}}

\def\C{\ensuremath{\mathcal C}}

\def\F{\ensuremath{\mathcal F}}

\newcommand*\poly{\ensuremath{\mathrm {poly}}}
\def\C{\ensuremath{\mathcal C}}

\let\sm\setminus
\def\F#1{f_{#1}}
\def\Tg{\Delta}
\def\Two{\Delta_{1}}
\def\Tas{\Delta_{\mathrm{as}}}
\def\Ts{\Delta_{\mathrm{s}}}
\def\bl{\bar\ell}
\def\BP{\mathop{\mathrm{NBP}}\nolimits}
\def\mBP{\mathop{\mathrm{NBP}_{+}}\nolimits}
\def\Eq{\mathop{\mathrm{Eq}}\nolimits}
\def\XOR{\mathop{\mathrm{XOR}}\nolimits}

\def\PT{\mathop{\mathrm{P3f}}\nolimits}

\begin{document}
\maketitle
\begin{abstract}
  A recently introduced  measure of Boolean functions complexity--disjunc\-tive complexity (DC)--is compared with other complexity measures: the space complexity of streaming algorithms and the complexity of nondeterministic branching programs (NBP).
  We show that DC is incomparable with NBP. Specifically, we present a function that has low NBP but has subexponential DC.
  Conversely, we provide arguments based on computational complexity conjectures to show that DC can superpolynomially exceed NBP in certain cases. Additionally, we prove that the monotone version of NBP complexity is strictly weaker than DC.

  We prove that the  space complexity of one-pass streaming algorithms is strictly weaker than DC. Furthermore, we introduce a generalization of streaming algorithms that captures the full power of DC.
This generalization can be expressed in terms of nondeterministic algorithms that irreversibly write 1s to entries of a Boolean vector (i.e., changes from 1 to 0 are not allowed).

Finally, we discuss an unusual phenomenon in disjunctive complexity: the existence of \emph{uniformly hard functions}. These functions exhibit the property that their disjunctive complexity is maximized, and this property extends to all functions dominated by them.
\end{abstract}

\pagestyle{plain}

\section{Introduction}

In this paper,
we study a novel complexity measure for Boolean functions, which we refer to as \emph{disjunctive complexity}.
A function $f\colon\{0,1\}^n\to \{0,1\}$
is represented succinctly by describing its accepting set  $f^{-1}(1)$ through paths in a labeled directed graph. We call these graphs \emph{transition graphs}. At first glance, this representation resembles the branching programs model (see~\cite{Jukna}). However, as we discuss in detail in Section~\ref{sect:BP}, the two models are  quite different. From another point of view,  disjunctive complexity generalizes the space complexity of streaming algorithms (see~\cite{Mu05}). In Section~\ref{sec:streaming}, we discuss this connection and propose several variants of disjunctive complexity to capture both the similarities and differences between the space complexity of streaming algorithms and our model.

Disjunctive complexity was introduced in~\cite{RV24}
as a tool for investigating the regular realizability problem in Formal Language Theory. In that work, the measure was defined for families of subsets of a finite set~$U$.
It is straightforward to observe that such families correspond naturally to Boolean functions: the indicator function of a subset is essentially a $(0,1)$-assignment of Boolean variables indexed by $U$, therefore a family of subsets corresponds to a Boolean function. 

In~\cite{RV24} several basic results on disjunctive complexity were obtained.  Notably,  lower bounds for this measure are relatively easy to prove, and an explicit class of functions with exponential disjunctive complexity was presented in the same work. Our main motivation behind this work is to deepen our understanding of the relationships between  disjunctive complexity and other complexity measures for Boolean functions. We anticipate that disjunctive complexity will find applications in  Boolean functions analysis.

\subsection*{Our Contribution}

We explore the relationship between streaming algorithms for computing Boolean functions and the model of representing Boolean functions using transition graphs. The most important characteristic of a streaming algorithm is its space complexity. Lower bounds of the space complexity are well-elaborated. They are based on communication complexity, see e.g. \cite{Jukna,AMS96,GT01,CGV20},  or on information-theoretic arguments, see~\cite{BJKS02}.

In Section~\ref{sec:streaming}, we demonstrate that streaming algorithms can be viewed as a special case of the transition graph representation model for Boolean functions. This connection provides a pathway to prove lower bounds on the space complexity of streaming algorithms by establishing lower bounds on disjunctive complexity. In the latter case, combinatorial arguments can be employed. In~\cite{RV24}, the separability property was used to derive such lower bounds. Here, we introduce a new tool based on the closure properties of the labeling functions in transition graphs. This tool yields subexponential lower bounds for very simple functions, such as the indicator function of $P_3$-free graphs (see Section~\ref{sect:BP} for the proof) or even simpler functions like the indicator function of cliques (see Section~\ref{sec:uniform-hard}). These results suggest that transition graphs may serve as an useful tool for proving lower bounds on the space complexity of streaming algorithms.

Up to a polynomial blow-up, the transition graph representation model can be interpreted as a way of generating the  \emph{accepting set} of a Boolean function $f$ (i.e. the set $f^{-1}(1)$) using a non-deterministic algorithm that writes units to a Boolean vector entries irreversibly. Comparing this model with streaming algorithms leads to several restricted variants of disjunctive complexity, which we present in Section~\ref{sec:streaming}. We prove two separation results for these restricted measures (see Theorems~\ref{stream<a.stream} and~\ref{write-once<general} in Section~\ref{sec:streaming}).

In Section~\ref{sect:BP}, we compare branching program complexity and disjunctive complexity. Assuming the complexity-theoretic conjecture $\NP\nsubseteq \P/\poly$, these complexity measures appear to be incomparable. We provide corresponding examples in this section. In one direction the separation is unconditional.

In Section~\ref{sec:uniform-hard}, we discuss an  interesting  phenomenon specific to disjunctive complexity. We prove that there exist functions  $f$ such that for all functions $g$ dominated by $f$, the disjunctive complexity of \( g \) is maximized and equals the size of $g^{-1}(1$). This phenomenon is specific to disjunctive complexity and has no known counterpart in other common complexity measures for Boolean functions.

\section{Definitions and Initial Observations}\label{sec:defs}

We allow directed graphs to have loops and multiple edges.
A \emph{transition graph} is a~tuple $G=(V,s,t,E,\ell)$, where
 $V$ is  the set of vertices, $E$ is the set of directed edges,  $s$ is the initial vertex, $t$ is the terminal vertex, and $\ell\colon E \to 2^U$ is the labeling function, where $U$ is a finite 
set, which we subsequently use as the index set for Boolean variables in the definition of the corresponding Boolean function.
The size $|G|$ of a transition graph is defined as the number of edges.
The set $\ell(e)$, $e\in E$, is called the \emph{edge label}. The \emph{label set}  $S = \ell(P)$ of a path $P$  is the union of all edge labels along the path. We also say that $P$ \emph{is marked} by~$S$.

A Boolean function  $\F{G}(x)$ described by $G$ is defined as follows. The input Boolean vector $x$ has dimension $|U|$ with components (variables) $x_i$, $i\in U$. An input vector $x$ has the corresponding \emph{unit set} $U_x = \{i\in U: x_i = 1\}$, i.e. the set of variables taking the value~1 in $x$.
So, $\F{G}(x)=1$ if the corresponding unit set is the label set of a path from  $s$ to $t$ (referred to as an $(s,t)$-path for brevity). Formally,
\[
\F{G}(x) = 1 \quad\Longleftrightarrow\quad \exists P : U_x = \ell(P), \
\text{where $P$ is an $(s,t)$-path.}
\]
We say that a transition graph $G$ \emph{represents} the function $\F{G}$. 
Graphs $G'$ and $G''$ are called \emph{equivalent} if $\F{G'} = \F{G''}$.

Our main concern are bounds on the sizes of transition graphs up to  a~polynomial blow-up. 
 At this level of precision, it suffices to focus on specific classes of transition graphs. For instance, only vertices lying on $(s,t)$-paths affect $ \F{G} $. Therefore, vertices that are unreachable from $ s $ or \emph{dead} (defined as vertices from which $ t $ is unreachable) can be removed from the graph. This yields an equivalent transition graph without unreachable or dead vertices. Note that this transformation can be performed in time polynomial in the size of the transition graph.

 Another example is restricting the labeling function to the empty label and \emph{singletons} ($1$-element sets). An edge $ (u_0, u_n) $ marked by a nonempty set $ S $ of size $ n $ can be replaced by a path $ u_0, u_1, \dots, u_n $, where $ u_1, \dots, u_{n-1} $ are new vertices, and each edge $ (u_i, u_{i+1}) $ is marked by a distinct singleton from $ S $. It is clear that this transformation produces an equivalent transition graph and takes a time polynomial in $|U|$, $|G|$.

An important result from \cite{RV24} is the following:

\begin{prop}[\cite{RV24}]
 For any transition graph $G$, there exists an equivalent directed acyclic graph (DAG) $G'$, which can be constructed in time polynomial in $ |U| $, $ |G| $.
\end{prop}

Based on this proposition, we will henceforth assume that a transition graph is always a DAG.
The \emph{disjunctive complexity} $\Tg(f)$, where $f\colon \{0,1\}^U\to \{0,1\}$, is the
smallest size of a~transition graph (DAG) representing $f$.
It is easy to see that
\begin{equation}\label{uniform-upbnd}
  \Tg(f)\leq |f^{-1}(1)|.
\end{equation}
This bound is attained by constructing a transition graph with edges $ e_x = (s, t) $ for each $ x $ such that $ f(x) = 1 $, where each edge $ e_x $ is marked by the unit set $ U_x $.

\section{Disjunctive Complexity and Streaming Algorithms}\label{sec:streaming}

In this section, we introduce several classes of labeling functions such that the corresponding transition graphs simulate computations performed by streaming algorithms. These classes give rise to restricted variants of disjunctive complexity.

In a streaming algorithm, input data is accessed in small portions that arrive sequentially, forming a ``stream''. The algorithm uses limited memory to compute a function of the input data. Boolean functions can be computed by such algorithms, where the most natural choice for a portion of input is a single bit.
We are interested in the space complexity of streaming algorithms. Therefore we represent a (non-deterministic) one-pass streaming algorithm $\mathcal{A}$ computing a Boolean function $f(x_1, \dots, x_n)$ in memory  $S$ by a transition relation $\delta\subseteq Q\times Q\times \{0,1\}$, where  $Q$ is the state set,  $|Q| = 2^S$. The algorithm starts a run from the initial state $q_0$. A valid run on the input
$x_1,\ldots,x_n$, $x_i\in\{0,1\}$ is a sequence of states $q_0, q_1, \ldots, q_n$, such that $(q_i, q_{i+1}, x_{i+1})\in \delta$ for all $0\leq i <n$. The value $f_{\A}(x_1,\dots, x_n)$ of the function computed by $\A$ is~$1$ if there exists a valid run $q_0, q_1, \ldots, q_n$ on the input $x_1,\ldots,x_n$ such that $q_n\in Q_a$, where $Q_a$ is the set of accepting states.
In this form, streaming algorithms resemble automata with $2^S$ states. However, for streaming algorithms,  the memory size may depend on $n$.

We have described the simplest model of streaming algorithms. In more general models, multiple passes through the input data are allowed. In the most general case, the variable to be read at a particular moment is determined by the current state of the algorithm. Such algorithms
correspond precisely to general transition graphs. 
Formally, a \emph{generalized streaming algorithm} $\A$  is determined by the following data:  the state set $Q$, 
the initial state is $s$, the halting state is $t$, the set of Boolean variables $U$; and the transition relation $\delta\subseteq Q\times Q\times\big(\{\emptyset\}\cup U\big)$.

The algorithm $\A$ operates on a Boolean vector from $\{0,1\}^U$, which is initially zero. Configurations of $\A$ are pairs $(q,x)$, where $q\in Q$, $x\in \{0,1\}^U$. The initial configuration is $(s,\vec{0})$; $\vec{0} = (0,\dots,0)$. 
A transition can change a configuration $(q, x)$ to a~configuration $(q', x')$ as follows. A transition $(q,q',\emptyset)\in \delta$ preserves $x$ (i.e., $x'=x$); a transition $(q, q', i)\in \delta$ sets $i$-th component to $1$: i.e., $x'_i =1$ while $x'_j = x_j$ for $j\ne i$. 
In a halting configuration $(t, x)$, the algorithm stops, and we say that the resulting vector $x$ is \emph{generated by} $\A$. The Boolean function $\F{\A}$ computed by $\A$ takes 
the value $1$ at all vectors generated by the algorithm and $0$ otherwise.

The transition graph of the algorithm $\A$ is $G = (V,s,t,E,\ell)$, where $V= Q$, $E$ consists of edges $e_{q',q'', i}$ such that $(q',q'', i)\in\delta$, and $\ell(e_{q',q'', i})$ is $i$. It follows immediately from this construction that that $\F{\A} = \F{G}$. 

Restricted types of streaming algorithms are represented by transition graphs with restricted labeling functions. We consider three types of restricted disjunctive complexity:

A transition graph is  \emph{streaming}  if  the variables in $U$
are linearly ordered (without loss of genarality (w.l.o.g.) we assume that
$U = [n] = \{1,2,\dots, n\}$)\footnote{Hereinafter we adopt notation $[n] = \{1,2,\dots, n\}$.}, edge labels  are singletons or the empty label, and along any $(s,t)$-path non-empty labels form an increasing sequence. For $f\colon \{0,1\}^{[n]} \to \{0,1\}$  the \emph{streaming disjunctive complexity} $\Ts(f)$ is the smallest size of a streaming transition graph  representing $f$.

A transition graph is    \emph{adaptively streaming} if there exists a linear order on the Boolean variables such that the graph  is  streaming with respect to this order.  For $f\colon \{0,1\}^U \to \{0,1\}$ the \emph{adaptive streaming disjunctive complexity} $\Tas(f)$ is the smallest size of an adaptively streaming transition graph  representing $f$.

A transition graph is \emph{write-once} if, along any $(s,t)$-path, each variable appears at most once. For $f\colon \{0,1\}^U \to \{0,1\}$ the \emph{write-once disjunctive complexity} $\Two(f)$ is the smallest size of a write-once transition graph  representing $f$.

It is clear from the definitions that $\Ts(f)\geq \Tas(f)\geq \Two(f)\geq \Tg(f)$.
We present below  an exponential separation between $\Tas$ and $\Ts$. We are able to prove a superpolynomial separation between $\Two$ and $\Tg$ under the conjecture $\NP\nsubseteq \P/\poly$. The  value of a gap between  $\Two(f)$ and $\Tas(f)$   remains open.

The separation $\Tas$ from $\Ts$ uses the well-known idea of deriving lower bounds for streaming space complexity by applying lower bounds from communication complexity. We adapt this idea to the case of disjunctive complexities.

\subsection{Streaming Disjunctive Complexity}

The conversion of a one-pass streaming algorithm to a streaming transition graph is straightforward.

\begin{prop}\label{stream-alg2stream-graph} 
If there exists a streaming algorithm computing $f\colon \{0,1\}^n\to\{0,1\}$ in memory $m(n)$, then $\Ts(f) = O(n2^{m(n)})$.
\end{prop}
\begin{proof}
  The operation of a streaming algorithm computing a Boolean function $f$ can be represented by a streaming transition graph $G = (V, s, t, E, \ell)$ as follows. Let $Q$ be the state set of the algorithm. The vertex set $V$ consists of pairs $(q, i)$, where $q \in Q$ and $0 \leq i \leq n$, along with the terminal vertex $t$. The initial vertex is $s = (q_0, 0)$, where $q_0$ is the initial state of the algorithm.
For each triple $(q',q'', x)$ in the transition relation of the algorithm, there are edges in $E(G)$ in the form  $e_{q',q'',i} =\big((q', i-1), (q'',i)\big)$ for all $1 \leq i \leq n$. The label of $e_{q',q'',i}$ is $\{i\}$ if $x=1$, otherwise the label is the empty set.

 This construction ensures that the label set of an $(s,t)$-path in $G$ is exactly the unit set of a vector $x$ such that $f(x)=1$. Conversely, if $f(x)=1 $, the streaming computation on the input $x$ produces an $(s, (q,n))$-path $P$, where $q$ is an accepting state. This path can be extended to the $(s,t)$-path $P, ((q,n),t)$ and  the label set of this path is exactly the unit set of $x$. Thus, $f = \F{G}$ and $|G| = (n+1) |Q| =O(n2^{m(n)}) $. 
\end{proof}

Using streaming algorithms and this proposition, we can immediately construct examples of functions with low disjunctive complexity.

\begin{example}\label{upbndXOR}
  The simplest example is the parity function  $\XOR_n(x_1,\dots, x_n)$. It takes the value $1$ if the number of variables set to $1$ is odd. To compute $\XOR_n$ by a streaming algorithm, only $1$ bit of memory is needed to track the parity of 1s encountered so far.  Thus $\Ts(\XOR_n) = O(n)$ due to Proposition~\ref{stream-alg2stream-graph}. 
\end{example}

\begin{example}
  The threshold function $T_{n,k}(x_1,\dots, x_n)$ is equal to $1$ if at least $k$ variables are set to $1$.
  The well-known majority function is  a special case of the threshold function, namely, $T_{n, \lfloor n/2 \rfloor + 1}$. The slice function $S_{n,k}(x_1,\dots, x_n)$ is equal to $1$ if exactly $k$ variables  are set to~$1$.
  Streaming algorithms for computing $T_{n,k}$ and $S_{n,k}$ store the count of 1s encountered so far and make a decision at the end of the stream.
  For this purpose they use a counter of encountered 1s of size $\lceil\log_2n\rceil$. So, by Proposition~\ref{stream-alg2stream-graph},  $\Ts(T_{n,k}) = O(n^2)$ and $\Ts(S_{n,k}) = O(n^2)$.
\end{example}

\begin{example}
  More generally, counting the number of $1$s suffices to compute any symmetric function. Thus, for any symmetric function $f\colon \{0,1\}^n \to \{0,1\}$, we have $\Ts(f) = O(n^2)$.
\end{example}

For a  transition graph $G= (V,s,t,E,\ell)$, we define the counting function $\#_1\colon V\to \ZZ$ as $\#_1(v) = \max(|\ell(P)|: P\ \text{is an $(s,v)$-path})$. Note that $\#_1(u)\leq \#_1(v)$ for any edge $e=(u,v)\in E$.

\begin{theorem}\label{stream<a.stream}
  There exists a sequence of Boolean functions $f_n\colon \{0,1\}^{4n}\to\{0,1\}$ such that $\Tas(f_n) = \poly(n)$ and $\Ts(f_n) = 2^{\Omega(n)}$.
\end{theorem}
\begin{proof}
  The functions $\Eq_n\colon \{0,1\}^{4n} \to \{0,1\}$ are defined as follows. The variable set is $[2n] \times [2]$. We assume that the input variables are ordered as: $(1,1)<(2,1)<\dots <(2n,1)<(1,2)<(2,2)<\dots< (2n,2)$.
  A binary word $x$ consists of the values of variables $(i,1)$, and a binary word $y$ consists of the values of variables $(i,2)$. By definition, $\Eq_n(x, y) = 1$ iff the Hamming weight of $x$ is $n$ and $x = y$ (i.e., $x_i = y_i$ for all $1 \leq i \leq 2n$).

    First, we prove that $\Tas(\Eq_n) = \poly(n)$. Consider the alternating order on the variables: $(1,1)<(1,2)< (2,1)<(2,2) < \dots< (2n,1)< (2n,2)$.   A~streaming algorithm computing $\Eq_n$ for this order is straightforward and uses $O(\log n)$ memory. The algorithm checks the conditions $x_i = y_i$ for all $1 \leq i \leq 2n$, counts the number of $1$s in $x$, and accepts if and only if all checks are successful and the total number of $1$s in $x$ is $n$. This algorithm can be converted into a~streaming transition graph representing $\Eq_n$. The size of the graph is $O(n^2)$ due to Proposition~\ref{stream-alg2stream-graph}.

  Next, we prove that $\Ts(\Eq_n) = 2^{\Omega(n)}$.  Let $G= (V,s,t,E,\ell)$ be a streaming transition graph representing $\Eq_n$.   Take an $(s,t)$-path $P$ in $G$ going through  a vertex $v\in V$ such that $\#_1(v)=n$. Divide the path into two subpaths $P'P''$, where $P' $ is an $(s,v)$-path and $P''$ is a $(v,t)$-path. From the definition of $\Eq_n$ we conclude  that $\ell(P')\subset \{(i,1): 1\leq i\leq 2n\}$ and $(i,2)\in \ell(P'')$ iff $(i,1)\in \ell(P')$.  Thus, $\ell(\widetilde P)= \ell(P')$ for any $(s,v)$-path  $\widetilde P$. Since $G$ represents $\Eq_n$, there are $\binom{2n}n$ different values of $\ell(P')$, each corresponding to a different vertex in the graph. Therefore, $|G|\geq \binom{2n}n = 2^{\Omega(n)}$.
\end{proof}

\subsection{Write-once Disjunctive Complexity}

\begin{lemma}\label{WOinP}
  There exists a polynomial time algorithm that, given a transition   graph $G$ and an assignment $x\colon U\to \{0,1\}$, decides whether $G$ is a write-once transition graph and, if so, computes $\F{G}(x)$.  
\end{lemma}
\begin{proof}
  We assume that the graph $G= (V,s,t,E,\ell)$ is DAG  and that the edge labels are either singletons or the empty label.  The assumption is valid because the transformation to this form can be performed in polynomial time, as explained in Section~\ref{sec:defs}.

  To verify that $G$ is write-once, the algorithm examines all pairs of edges $e_1$, $e_2$ such that $\ell(e_1)=\ell(e_2)\ne \es$. For each such pair, the algorithm checks  whether there exists no $(s,t)$-path that passes through both $e_1$ and $e_2$. This check is easily reduced to the reachability problem in a directed graph.

  Let $S = x^{-1}(1)$. The algorithm removes all edges $e$ in $G$ such that $\ell(e)\nsubseteq S$,  resulting in a modified graph $G' = (V, s, t, E', \ell)$.  It follows immediately from the definitions that $\F{G}(x) =1$ if and only if there exists an $(s,t)$-path in $G'$ labeled by $S$. In terms of the counting function of $G'$ it means that $\#_1(t) = |S|$. So, to compute $\F{G}(x)$, it suffices to compute $\#_1(t)$ in the graph $G'$.

  Since $G'$ remains write-once, the counting function satisfies a recurrence
  \begin{equation}\label{eq:cnt}
    \begin{aligned}
      &\#_1(s) = 0;\\
      &\#_1(v) = \max_{u: (u,v)\in E'}\big( \#_1(u)+|\ell(u,v)|\big).
    \end{aligned}
  \end{equation}
  To solve this recurrence, the algorithm first computes a topological sorting of the vertices in $G'$, which can be done in polynomial time. Then, it evaluates the recurrence using Eq.~\eqref{eq:cnt} in the order of the topological sorting, which also takes polynomial time.
\end{proof}

\begin{theorem}\label{write-once<general}
  Assume that $\NP\nsubseteq \P/\poly$. Then there exists a sequence of Boolean functions $f_n\colon \{0,1\}^{\poly(n)}\to\{0,1\}$ such that $\Tg(f_n) = \poly(n)$ and $\Two(f_n)$ grows faster than any polynomial in $n$.
\end{theorem}
\begin{proof}
  It was proved in~\cite{RV24} that computation $\F{G}(x)$ for general transition graphs is $\NP$-complete.\footnote{For reader's convenience we add the proof to Appendix~\ref{app:Npcomplete}.} If, for any general transition graph $G$, there exists an equivalent write-once transition graph of size $\poly(|G|)$, then there would exist Boolean circuits computing $\F{G}(x)$ with size polynomial in $|G|$. Here we apply Lemma~\ref{WOinP} and the well-known conversion of algorithms to Boolean circuits; see,  e.g.~\cite{AroraBarak}.
\end{proof}

\section{Disjunctive Complexity and Branching Programs}\label{sect:BP}

A (nondeterministic) branching program is a directed acyclic graph $G$ with edges labeled by \emph{literals} (variables and negations of variables). An assignment $x$ of variables switches on edges labeled by literals that are true under $x$. Switched on edges form an \emph{assignment graph} $G_x$. A Boolean function $f$ computed by a~branching program $G$ takes the value $1$ if   there exists an $(s,t)$-path in $G_x$.
The corresponding complexity measure $\BP(f)$ is the minimum size of a~branching program computing a Boolean function $f$. 
If all labels in a branching program are variables (i.e., no negations are used), then it is called \emph{monotone}. For a monotone Boolean function $f$, another complexity measure  $\mBP(f)$ is well-defined: it is the minimum size of a monotone branching program computing $f$. 

For more details on branching programs, see, e.g.,~\cite{Jukna}.

At first glance, branching programs look similar to transition graphs. However, these models are quite different.  For instance, transition graphs use  only variables (no negations) but can represent non-monotone functions (in fact, any function can be represented by a transition graph).  Moreover, a monotone branching program can be simulated by a transition graph with only a polynomial blow-up in size.

\begin{prop}
  Let $G$ be a monotone branching program computing a function~$f$. Then there exists a transition graph $G'$ such that $\F{G'}=f$ and $|G'|\leq |G|+O(n)$, where $n$ is the number of variables.
  Thus $\Delta(f)\leq \mBP(f)+O(n)$.
\end{prop}

\begin{proof}
  We assume that  variables of $f$ are indexed by $[n]$. Since there are no negations of variables, we use the same labeling function as before to describe $G$. The transition graph $G'$ is constructed from $G$ by adding new vertices $v_1, v_2,\dots, v_n$ and edges $e_{i,\al} = (v_{i-1}, v_i)$, where $1\leq i\leq n$, $\al\in\{0,1\}$, $\ell'(e_{i,0})=\es $, $\ell'(e_{i,1}) = \{i\}$. For uniformity, we set $t= v_0$, where $t$ is the terminal vertex in $G$. The initial vertex of $G'$ is $s'=s$, the terminal vertex is $t' = v_n$. The construction ensures that for each assignment $x\colon [n]\to \{0,1\}$ there exists a~$(v_0,v_n)$-path $P$ such that $\ell(P) = x^{-1}(1) $.  

  If $f(x)=1$, then there exist an $(s,t)$-path  in $G$ with edges labeled by $i$ such that $x_i=1$. The path  can be extended by a suitable $(v_0, v_n)$-path  to an $(s', t')$-path $P'$  such that $ \ell'(P') = x^{-1}(1)$. Therefore $\F{G'}(x) =1$.

 If $f(x)=0$, then for any an $(s,t)$-path $P$ in $G$, there exists an edge $e\in P$ such that $\ell(e)\notin x^{-1}(1)$. Thus, any  $(s',t')$-path in $G'$ contains an edge labeled by $i\notin x^{-1}(1)$. Therefore    $\F{G'}(x) =0$.  
\end{proof}

Given a (general) branching program $G$ computing a function $f$, the value $f(x)$ can be computed in polynomial time by solving a reachability problem in the assignment graph $G_x$. Since computing a value of a function represented by a transition graph is $\NP$-complete (see~\cite{RV24}), it leads to an analogue of Theorem~\ref{write-once<general}.

\begin{prop}\label{BP<general}
   Assume that $\NP\nsubseteq \P/\poly$. Then there exists a sequence of Boolean functions $f_n\colon \{0,1\}^{\poly(n)}\to\{0,1\}$ such that $\Tg(f_n) = \poly(n)$ and $\BP(f_n)$ grows faster than any polynomial in $n$.
\end{prop}

The proof of this proposition repeats the proof of Theorem~\ref{write-once<general} with appropriate modifications.

In  the opposite direction, we can prove unconditional separation of branching and disjunctive complexities.
For this purpose, we introduce Boolean functions with variables indexed by unordered pairs of integers in $[n] = \{1, \dots, n\}$.  An assignment  $x\colon\binom{[n]}2\to\{0,1\} $ in this case is naturally interpreted as an undirected graph $\Gamma_x$ on the vertex set $[n]$: edges of $\Gamma_x$ are pairs $\{i,j\}$ such that $x_{ij} =1$. To distinguish these graphs from transition graphs we call them \emph{label graphs}. A~Boolean function $f\colon\{0,1\}^{\binom{[n]}2}\to \{0,1\}$ is interpreted as a~property of graphs on $n$ vertices.

We now define functions that separate branching and disjunctive complexities. Let $ \PT_n$ be the indicator function of $P_3$-free graphs.  That is, $U =\binom{[n]}2 $ and $\PT_n(x)=1$ if $\Gamma_x$ is
a $P_3$-free graph, i.e. $\Gamma_x$ does not contain an induced  path on 3 vertices.
It is well-known that $P_3$-free graphs are disjoint unions of complete graphs.

\begin{prop}
  $\BP(\PT_n) = O(n^3)$.
\end{prop}
\begin{proof}
  A graph is $P_3$-free if and only if, for any triple $v_1, v_2, v_3$ of its vertices, the induced subgraph does not contain exactly two edges. In other words, the existence of two edges between these vertices implies the existence of the third one. Since $(x\land y)\mathbin{\to} z$ is equivalent to $\lnot x \lor \lnot y \lor z$, this condition is easily implemented in a branching program. We construct a branching program $G = (V, s,t,E, \ell)$ as follows. The input variables are unordered pairs $\{i,j\}$, $1\leq i\ne j \leq n$.   Let $[n]_3$ be the set of ordered arrangements of triples of integers in the range $[n]$,
  and $\nu\colon [N]\to [n]_3$ be an enumeration of $[n]_3$, here $N=n(n-1)(n-2)$.
  Set $V= \{0\}\cup [N]$, $s = 0$, $t = N$. For each $1\leq i\leq N$ there are three edges $e_{i,\al} = (i-1,i)$, $\al\in\{1,2,3\}$.  Their labels are $\ell(e_{i,1}) =\lnot x_{\{p,q\}}$, $\ell(e_{i,2}) =\lnot x_{\{q,r\}}$, $ \ell(e_{i,3}) = x_{\{r,p\}}$, where $\nu(i) = (p,q,r)$. Here, we extend the labeling function to allow literals as labels.   An $(s,t)$-path in $G$ does exist if and only if
there exists $x$ such that $\lnot x_{\{p,q\}}\lor\lnot x_{\{q,r\}}\lor x_{\{r,p\}} $ is true for all $p,q,r$,
which is equivalent to  $\Gamma_x$ being $P_3$-free. Thus, $G$ computes $\PT_n$. 
\end{proof}

To prove a lower bound for $\Tg(\PT_n)$ we need to show that  a restricted class of labeling functions is equivalent to the general case.
Let $G=(V,s,t,E,\ell)$ be a transition graph. The \emph{closure} $\bl$ of the labeling function $\ell$ is defined as follows. For   $e= (u,v)\in E(G)$ the label set  $\bl(e)$ is the union of three sets: (1) the label set $\ell(e)$; (2) the set of    $a\in U$ such that  $a\in \ell(P)$ for any  $(s,u)$-path~$P$;
(3) the set of  $a\in U$ such that $a\in \ell(P)$ for any $(v,t)$-path~$P$.

\begin{lemma}\label{lm:equivalence-closure}
 The closure of the labeling function yields an equivalent transition graph: $\F{(V,s,t,E,\ell)} = \F{(V,s,t,E,\bl)} $.
\end{lemma}

\begin{proof}
  Let $P$ be an  $(s,t)$-path.  By definition, $\ell(e) \subseteq \bl(e)$ for any edge $e$, so  $ \ell (P)\subseteq \bl(P)$. Now we prove the opposite inclusion. Let $a\in \bl(e)$, $e = (u,v)\in P$. The edge $e$ divides $P$ in three parts: $P =P' e P'' $.

  If $a\in \ell(e)$, then $a\in \ell(P)$. If $a\in \bl(e)\sm\ell(e)$, then, by definition, at least one of the following holds:  $a\in \ell(P')\subseteq \ell(P)$ or  $a\in \ell(P'')\subseteq \ell(P)$. Thus $\bl(P)\subseteq \ell(P)$.

We have proved that  $\bl(P) =  \ell(P)$, which implies that the two labeling functions yield equivalent transition graphs.
\end{proof}

We say that the labeling function $\ell$ is closed if $\ell = \bl$. From Lemma~\ref{lm:equivalence-closure}, we conclude that closed labeling functions determine the disjunctive complexity. 

\begin{lemma}\label{lm:cl2=cl}
   The closure operator is idempotent: the closure of the closure of $\ell$ is the closure of $\ell$.
\end{lemma}
\begin{proof}
  Let   $a\notin\bl(e)$. It implies that  $a\notin \ell(e)$, and there exists a path $P =P' e P''  $ such that  $a\notin \ell(P')$ and $a\notin \ell(P'')$. So   $a\notin \bl(e')$ for any edge $e'$ of  $P$, since $e'$ divides  $P$ in three parts satisfying the same conditions. Therefore $a$ does not belong to the closure of $\bl$.
\end{proof}

To get the desired lower bound, we need to establish several properties of transition graphs representing functions whose accepting sets are subsets of the set of $P_3$-free graphs.
 Fix a transition graph $G =(V,s,t,E,\ell)$ such that   $U= \binom{[n]}2$ and  $\F{G} (x)= 1$ implies that $\Gamma_x$ is $P_3$-free. Due to Lemma~\ref{lm:equivalence-closure}, we assume w.l.o.g. that $\ell$ is closed.

\begin{prop}\label{l(e)-P3-free}
  For any $e\in E$  the label graph  $\ell(e)$ is $P_3$-free.
\end{prop}

\begin{proof}
  Suppose for contradiction that $\{i,k\}\in \ell (e)$, $\{j,k\}\in \ell(e)$, and $\{i,j\}\notin \ell(e)$. Since $\ell$ is closed, due to Lemma~\ref{lm:cl2=cl}, there exists an $(s,t)$-path $P = P' e P''$ such that  $\{i,j\}\notin \ell(P')$ and  $\{i,j\}\notin \ell(P'')$. Thus the label graph
  $\ell(P)$ is not  $P_3$-free,
  contradicting the assumption made.
\end{proof}

\begin{prop}\label{P3-free-path}
  For any  $(s,x)$-path $P$, the label graph $\ell(P)$ is $P_3$-free. Moreover, for any  $(s,x)$-path  $P$ and any edge $(x,y)\in E$ the maximal cliques of $\ell(P)$ and $\ell(x,y)$  either do not intersect or are comparable under the inclusion relation.
\end{prop}
\begin{proof}
The induction on the length of $P$. The base case follows from Proposition~\ref{l(e)-P3-free}.

  The induction step. Let $C_1$ be a maximal clique in  $\ell (P)$ and $C_2$ be a maximal clique in $\ell(x,y)$.
Suppose for contradiction that $i\in C_1\sm C_2$,  $j\in C_2\sm C_1$, and $k\in C_1\cap C_2$.
Then  $\{i,j\}\notin \ell (P)\cup \ell(x,y)$.  Thus, there exists a $(y,t)$-path $P'$ such that  $\{i,j\}\notin \ell (P')$ due to Lemma~\ref{lm:cl2=cl}.  Therefore, the label graph of the path $P (x,y) P'$ is not  $P_3$-free, since  $\{i,j\}\notin \ell(P (x,y) P')$, while $\{i,k\}\in \ell(P (x,y) P')$ and $\{j,k\}\in \ell(P (x,y) P')$. We come to a contradiction to the assumption made.
\end{proof}

\begin{theorem}\label{P3-free-lwrbnd}
  $\Tg(\PT_n) = 2^{\Omega(n)}$. 
\end{theorem}
\begin{proof}
  Let $\F{G} = \PT_n$. Without changing the number of edges in $G$, we assume that the labeling function is closed. From Proposition~\ref{P3-free-path}, we conclude that for any $(s,t)$-path $P$ every
maximal  clique
  in  $\ell(P)$ must be a clique in some $\ell(e)$ for an edge $e$ in the path.
  Since
  the number of maximal cliques in any label graph $\ell(e)$ is at most $n$ (since it is $P_3$-free),
  and there are $2^n$ possible cliques, the number of edges in $G$ must be at least $2^n / n$.
\end{proof}

\begin{remark}
  It is easy to verify that  that $\lnot \PT_n$ is represented by a transition graph of  size polynomial in $n$.  The corresponding generalized streaming algorithm nondetermistically guesses a triple $i,j,k$, writes $x_{ij} =1$, $x_{jk}=1$, and after that writes arbitrary values for the remaining edges in the label graph except $\{i,k\}$.
  Thus, we get an example of a function such that the gap between the disjunctive complexity of the function and  its negation is  (sub-)exponential. 
\end{remark}

\section{Uniformly Hard Functions}\label{sec:uniform-hard}

A Boolean function $f$ is \emph{uniformly hard} with respect to the disjunctive complexity if $\Tg(g)=|g^{-1}(1)|$ for any function $g$ such that $g^{-1}(1)\subseteq  f^{-1}(1)$ (denoted as $g\leq f$). It means that the disjunctive complexity of $f$ is largest possible (see the bound~\eqref{uniform-upbnd}) and, moreover, the same holds for all functions that are dominated by~$f$, i.e. $g(x)\leq f(x)$ for all $x$.

In~\cite{RV24} exponential lower bounds on the disjunctive complexity were proved for indicator functions of asymptotically good binary linear codes (AGLC functions). These functions are \emph{almost uniformly hard}. It means that $|g^{-1}(1)| = \poly(\Tg(g))$ for all $g\leq f$. For AGLC functions, the known estimates of the degree of a~polynomial in this bound depend on the relative code distance and are rather large. 

Here we present an example of uniformly hard functions.  A graph is called a \emph{clique graph} if it is a~disjoint union of a complete graph and a set of isolated vertices. The name reflects the fact that  a clique graph is determined by its unique non-trivial clique, except in the case of the empty graph. 
The function $\C_n$ is an indicator function of clique graphs, i.e.  $\C_n(x) =1$ if $\Gamma_x$ is a clique graph.
By definition, $|\C_n^{-1}(1)| =2^n$. A modification of the  proof of Theorem~\ref{P3-free-lwrbnd} gives  the matching lower bound.

\begin{theorem}\label{cliques}
  $\Tg(\C_n) = 2^{n}$.
\end{theorem}

Fix a transition graph $G =(V,s,t,E,\ell)$ such that   $U= \binom{[n]}2$ and  $\F{G}\le \C_n$. W.l.o.g. we assume that $\ell$ is closed. 

\begin{prop}\label{l(e)-clique}
  For any $e\in E$  the label graph  $\ell(e)$ is a clique graph.
\end{prop}

\begin{proof}
 Let $i,j$ be vertices of positive degree in the label graph $\ell(e)$. Specifically, let $\{i,i'\}\in \ell(e)$ and $\{j,j'\}\in\ell(e)$. We will  prove that  $\{i,j\}\in \ell(e)$. This implies the proposition, since in any graph that is not a clique graph, there exists a pair of non-adjacent vertices of positive degrees.

  Suppose for contradiction  that  $\{i,j\}\notin \ell(e)$. Due to Lemma~\ref{lm:cl2=cl}, there exists an $(s,t)$-path $P = P' e P''$ such that  $\{i,j\}\notin \ell(P')$ and  $\{i,j\}\notin \ell(P'')$. Thus, the label graph $\ell(P)$ is not a clique graph, since  $\{i,i'\}\in \ell(P)$, $\{j,j'\}\in\ell(P)$. This contradicts the assumption  $\F{G}\le \C_n$.
\end{proof}

\begin{prop}\label{clique-path}
  For any  $(s,x)$-path $P$, the label graph $\ell(P)$ is a clique graph. Moreover, for any  $(s,x)$-path  $P$ and any edge $(x,y)\in E$,
  the maximal cliques of $\ell(P)$ and $\ell(x,y)$
  are comparable under the inclusion relation.
\end{prop}
\begin{proof}
The induction on the length of $P$. The base case follows from Proposition~\ref{l(e)-clique}.
For the induction step, if $\ell (P)$ or $\ell(x,y)$ are empty graphs, the claim follows from
Proposition~\ref{l(e)-clique} or from the induction hypothesis.
Otherwise, let  $\ell (P)$ be a clique graph with a non-trivial clique $ A$, where $P$ is an  $(s,x)$-path  $P$, and $\ell(x,y)$ be a clique graph with a~non-trivial clique $ B$, where $(x,y)\in E$. Suppose for  contradiction that   $A$ and $B$ are incomparable under the inclusion relation. It means that there exist $i\in A\sm B$ and $j\in B\sm A$. Thus,  $\{i,j\}\notin \ell (P)\cup \ell(x,y)$.  Due to Lemma~\ref{lm:cl2=cl}, there exists a~$(y,t)$-path $P'$ such that  $\{i,j\}\notin \ell (P')$.  Therefore, the label graph of the path $P (x,y) P'$ is not a clique graph, since  $i$, $j$ are non-adjacent in the graph but both have positive degree. This contradicts to the assumption $\F{G}\le \C_n$. 
\end{proof}

\begin{proof}[of Theorem~\ref{cliques}]
  Let $\F{G} = \C_n$. Without changing the number of edges in $G$, we assume that the labeling function is closed.  From Proposition~\ref{clique-path}, it follows that the label graph  $\ell(P)$ of any $(s,t)$-path $P$ coincides with the label graph of some edge in the path. Thus, the number of edges in $G$ is at least the number of satisfying assignments of $\C_n$, i.e. $2^n$.
\end{proof}

The existence of (almost) uniformly hard functions looks specific for the disjunctive complexity. For example, it is open whether almost uniform hard functions exist with respect to circuit complexity.

\newpage
\appendix

\section{An $\NP$-completeness of an Evaluation a Function Represented by a Transition Graph}\label{app:Npcomplete}

We reproduce the result from~\cite{RV24}.

\begin{theorem}\label{Gell-value}
Computing $\F{G}(x)$ is $\NP$-complete. 
\end{theorem}

\begin{proof}
  W.l.o.g. we assume that $G$ is a DAG. It implies that computing $\F{G}(x)$ is in $\NP$. An $\NP$-certificate is just an $(s,t)$-path in $G$ marked by the unit set of $x$.

  To prove hardness, we reduce SAT to this problem.  Let CNF $C = \bigwedge_{j=0}^{m-1} D_j$ be an instance of SAT.

  Define a graph and a labeling function as follows. The vertex set is the set  $X$ of variables $x_1,\dots, x_n$ occuring in~$C$ and the initial vertex $x_0$. The terminal is  $x_{n}$. There are two edges  $e_{i,0}$, $e_{i,1}$ from   $x_{i-1}$ to $x_{i}$, $0< i\leq n$. It is clear that the graph is constructed in polynomial time by $C$. 
The Boolean variable set is  $[m]$. The edge label $\ell(e_{i,\al})$ consists of  $j$ such that $D_j=1$ if   $x_i=\al$ (in other words, the clause $D_j$ contains the literal $x_i^{\al}$). 
Choose $x = (1,\dots,1)$, in other words, $x^{-1}(1)=[m]$.

  Note that $(s,t)$-paths in the graph are in one-to-one correspondence with assignments of variables in $C$. The label set of the $(s,t)$-path corresponding to an assignment $\al$ is the set of indexes of clauses that are satisfied by this assignment. Thus, the existence of an  $(s,t)$-path marked by  $[m]$ is equivalent to satisfiability of~$C$. It proves the correctness of the reduction.
\end{proof}


\begin{thebibliography}{8}
\bibitem{AMS96}   N. Alon, Y. Matias and M. Szegedy. The space complexity of approximating the frequency moments.
Proc. ACM STOC, 20–29, 1996.
\bibitem{AroraBarak} S. Arora and B. Barak, Computational Complexity: A Modern Approach. Cambridge Univ. Press, 2009.
\bibitem{BJKS02}  Z. Bar-Yossef, T. Jayram, R. Kumar and D. Sivakumar. Information statistics approach to data stream
and communication complexity. Journal of Computer and System Sciences
Volume 68, Issue 4,  Pages 702-732.
\bibitem{CGV20}  
  C.-N. Chou, A. Golovnev and S. Velusamy, Optimal Streaming Approximations for all Boolean Max-2CSPs and Max-$k$sat, 2020 IEEE 61st Annual Symposium on Foundations of Computer Science (FOCS), Durham, NC, USA, 2020, pp. 330-341, 
\bibitem{GT01}  P. Gibbons and S. Trithapura. Estimating simple functions on the union of data streams. ACM SPAA,
2001.
\bibitem{Jukna} S. Jukna, Boolean Function Complexity Advances and Frontiers. Springer-Verlag: Berlin Heidelberg, 2012.
\bibitem{Mu05}  S. Muthukrishnan, Data Streams: Algorithms and Applications, Foundations and Trends® in Theoretical Computer Science, 2005. Vol. 1: No. 2, pp 117-236.
\bibitem{RV24} Rubtsov A.A., Vyalyi M.N. On Universality of Regular Realizability Problems. Problems of Information Transmission, 2024, Vol. 60, No. 3, pp. 209–232. 

\end{thebibliography}
\end{document}